\documentclass{article}
\usepackage[utf8]{inputenc}
\usepackage{amsmath}
\usepackage{amsthm}
\usepackage{amssymb}
\usepackage{tikz}
\usetikzlibrary{arrows, automata}
\usetikzlibrary{intersections}
\usepackage{hyperref}
\usepackage{breakcites}
\usepackage{url}
\usepackage{cite}
\usepackage{authblk}
\usepackage{graphicx}
\usepackage{adjustbox}
\usepackage{caption}
\usepackage{subcaption}
\usepackage{appendix}
\usepackage{pgfplots}
\usepgfplotslibrary{fillbetween}
\pgfplotsset{compat=1.11}
\pgfdeclarelayer{bg}
\pgfsetlayers{bg,main}
\usepackage{algorithm}
\usepackage{algpseudocode}

\newtheorem{theorem}{Theorem}

\usepackage{lineno}

\providecommand{\subjclass}[1]{\textbf{2010 AMS Subject Class:} #1}
\providecommand{\keywords}[1]{\textbf{Keywords:} #1}
\title{A Linear Algorithm for Minimum Dominator Colorings of Orientations of Paths}
\author[1]{Michael Cary\footnote{macary@mix.wvu.edu}}
\affil[1]{Division of Resource Economics and Management, West Virginia University}
\begin{document}
\maketitle

\begin{abstract}
In this paper we present an algorithm for finding a minimum dominator coloring of orientations of paths. To date this is the first algorithm for dominator colorings of digraphs in any capacity. We prove that the algorithm always provides a minimum dominator coloring of an oriented path and show that it runs in $\mathcal{O}(n)$ time. The algorithm is available at \url{https://github.com/cat-astrophic/MDC-orientations_of_paths/}. 
\end{abstract}



\section{Introduction}
Let $G=(V,E)$ be a graph. A set $S\subset V$ is called a dominating set if every vertex of $V$ is either in $S$ or adjacent to at least one member of $S$. The domination number of a graph, $\gamma(G)$, is the size of a smallest dominating set of $G$. A dominator coloring of a graph is a proper vertex coloring of the graph which additionally satisfies the property that every vertex dominates some color class in the dominator coloring. Dominator colorings of graphs were first studied by Gera in \cite{gera2007a, gera2007b, gera2006dominator}. Dominating sets and dominator colorings are useful for a myriad of problems and have been applied to studying electric power grids \cite{haynes2002domination} and to sensors in networks \cite{blair2011movable}.

Since the purpose of this paper is to introduce an algorithm for minimum dominator colorings of digraphs, it is important to know what algorithms exist for this problem in the undirected setting. A linear algorithm which finds the domination number of trees was introduced in \cite{cockayne1975linear}. Building from this result, \cite{merouane2015} established a polynomial time algorithm that provides a minimum dominator coloring of trees. Additionally, it was proved in \cite{arumugam2011algorithmic} that the dominator chromatic number, $\chi_{d}(G)$, cannot be found in polynomial time in general for many elementary families of graphs including bipartite and planar graphs.

The first results on dominator colorings of digraphs were the dominator chromatic number of paths and cycles \cite{cary2019}, and that the dominator chromatic number of digraphs is that the dominator chromatic number of a tree is invariant under reversal of orientation \cite{cary2019dominator}. Other interesting results established in \cite{cary2019} on dominator colorings of digraphs include the fact that the dominator chromatic number of a subgraph $H\subset G$ can be larger than the dominator chromatic number of $G$, as well as that $\limsup\limits_{n\to\infty}\frac{\chi_{d}(D)}{\Delta(D)}=\infty$. That these results differ from virtually every other vertex coloring problem makes dominator colorings of digraphs of particular interest in graph theory.

One of the major results from that paper was the following theorem which provides the minimum dominator chromatic number over all orientations of paths. In the proof, many important structural characterizations relating orientations of paths and dominator colorings were established.

\begin{theorem}\label{t1}
The minimum dominator chromatic number over all orientations of the path $P_{n}$ is given by
\begin{equation*}
\chi_{d}(P_{n})=\begin{cases}
k+2 & \mathrm{if}\ n=4k\\
k+2 & \mathrm{if}\ n=4k+1\\
k+3 & \mathrm{if}\ n=4k+2\\
k+3 & \mathrm{if}\ n=4k+3
\end{cases}
\end{equation*}
for $k\geq 1$ with the exception $\chi_{d}(P_{6})=3$.
\end{theorem}

To conclude the introduction, an example of an oriented path of length five is presented below, and a minimum dominator coloring is provided in the caption. The reader is referred to \cite{haynes2013fundamentals} as the standard reference for domination.

\begin{figure}[h!]
\centering
\begin{tikzpicture}[-,>=stealth',shorten >=1pt,auto,node distance=2cm,
                    thick,main node/.style={circle,draw}]
  \node[main node] (A)					      {$v_{1}$};
  \node[main node] (B) [right of=A]		      {$v_{2}$};
  \node[main node] (C) [right of=B] 	      {$v_{3}$};
  \node[main node] (D) [right of=C]      	  {$v_{4}$};
  \node[main node] (E) [right of=D]           {$v_{5}$};
  \draw[thick,->] (B) to [bend left] (A);
  \draw[thick,->] (B) to [bend left] (C);
  \draw[thick,->] (D) to [bend left] (C);
  \draw[thick,->] (D) to [bend left] (E);
\end{tikzpicture}
\caption{An orientation of the path $P_{5}$ which has dominator chromatic number 3. The color classes of $P_{5}$ in a minimum dominator coloring may be given by $C_{0}=\{v_{2},v_{4}\}$, $C_{1}=\{v_{1},v_{3}\}$, and $C_{2}=\{v_{5}\}$.}
\label{f4}
\end{figure}
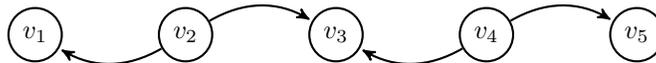

\section{The Algorithm}
The purpose of this paper is to present the first algorithm which provides a minimum dominator coloring of a directed graph. In particular, we present an algorithm which provides a minimum dominator coloring of orientations of paths. After providing the algorithm, we will prove that it runs in $\mathcal{O}(n)$ time.

The algorithm works by sequentially through the vertex set, from $v_{1}$ through $v_{n}$, and colors each vertex in such a a way as to minimize the number of colors used in a proper dominator coloring. From Theorem \ref{t1} in \cite{cary2019} we know that all vertices with in-degree equal to zero must belong to the same color class. For this reason, the first thing the algorithm checks for is precisely this. Assuming this is not the case, another immediately guaranteed coloring results is that any vertex that is dominated by a vertex of out-degree one must be uniquely colored. Once these two cases are checked for, all that remains in an orientation of a path are vertices whose entire in-neighborhood consists of vertices with out-degree equal to two (notice that this may include vertices with out-degree zero or out-degree one).

It is easy to see that, for oriented paths, after removing all vertices with with at least one in-neighbor having out-degree one, what remains are subpaths in which every vertex has either out-degree zero or out-degree two, i.e., subpaths which have the following out-degree sequence pattern: $\{0,2,0,\dots,0,2,0\}$ (it is possible that one or both of the end vertices of such a subpath do not have out-degree zero, but since all of the out-degree two vertices are assigned the same color, and since such an end vertex would have out-degree two in the full path, we may ignore end vertices of these subpaths that do not have out-degree zero). From Theorem \ref{t1} in \cite{cary2019} we know that such oriented paths are minimized in terms of dominator colorings precisely when the first vertex is assigned a color $C^{\star}$ which is used for all vertices of out-degree zero that are not uniquely colored, and when this color is assigned to every other vertex with out-degree zero (every fourth vertex in the path). For convenience, we refer to these paths as 2-chains since they are maximal subpaths with respect to the density of vertices with out-degree two. Also, 2-chains of length three are an exception to this coloring scheme as they may have both vertices of out-degree zero belong to the same color class. This case is addressed in the algorithm as well, but first we describe the process for handling all other 2-chains.

By paying attention to where we are within a 2-chain in an oriented path, we can proceed in order through the vertex set of an oriented path and provide a minimum dominator coloring. The method used in the algorithm is to have a variable $\alpha\in\mathbb{Z}_{2}$ indicate whether or not we should use the color $C^{\star}$ or a new color when coloring vertices of out-degree zero. The variable $\beta$ indicates whether we have established the color $C^{\star}$ yet. Since it turns out that 2-chains of length three may use different colors for the two vertices of out-degree zero if one vertex is colored with $C^{\star}$ and $C^{\star}$ is present elsewhere in the path (this will be proven in the theorem below), the algorithm handles 2-chains of length three by coloring both vertices the same if and only if the color $C^{\star}$ has not yet been established (i.e., if $\beta $ still has a value of $0$).

Lastly we address the exception of the path $P_{6}$ which has either the out-degree sequence $\{0,2,0,2,0,1\}$ or its reversal. As is stands, the algorithm would minimally color the reversal of this particular orientation of $P_{6}$ but not this particular orientation of $P_{6}$. Because this was the only exception listed in Theorem \ref{t1} from \cite{cary2019}, and because checking an input ($n$) will not alter the time complexity of the algorithm, we simply check for this occurrence in the same \textbf{if} statement as when we check to see if a 2-chain of length three exists.

Finally we are ready to present the algorithm. After stating the algorithm, we provide its time complexity and prove that it always provides a minimum dominator coloring of an oriented path.

\begin{algorithm}[h!]
\footnotesize
\begin{algorithmic}[1]
\caption{Minimum Dominator Coloring Algorithm for Oriented Paths}
\State \textbf{input} An orientated path $P_{n}$ of length $n$
\State \textbf{initialize} $\mathcal{C}\gets\{C_{0}\}$
\State \textbf{initialize} $\mathcal{F} \gets \emptyset$
\State \textbf{initialize} $\alpha\gets 0$
\State \textbf{initialize} $\beta\gets 0$
\For{$v_{i} \in V(P_{n})$}
\If{$d^{-}(v_{i}) = 0$}
    \State Color $v_{i}$ with $C_{0}$
    \State $\mathcal{F} \gets \{\mathcal{F}\cup C_{0}\}$
\ElsIf{$\exists\ v_{j} \in N^{-}(v_{i})\ \mathrm{s.t.}\ d^{+}(v_{j}) = 1$}
    \State Color $v_{i}$ uniquely with a new color $C_{|\mathcal{C}|}$
    \State $\mathcal{C} \gets \{\mathcal{C}\cup C_{|\mathcal{C}|}\}$
    \State $\mathcal{F} \gets \{\mathcal{F}\cup C_{|\mathcal{C}|}\}$
    \State $\alpha \gets 0$ \Comment{This indicates the end of a 2-chain}
\ElsIf{$\alpha = 0$}
    \If{$\beta = 0$}
        \State Define new color $C^{\star} = C_{|\mathcal{C}|}$
        \State Color $v_{i}$ with color $C^{\star}$
        \State $\mathcal{C} \gets \{\mathcal{C}\cup C^{\star}\}$
        \State $\mathcal{F} \gets \{\mathcal{F}\cup C^{\star}\}$
        \If{$d^{+}(v_{i+1}) = 2 \neq d^{+}(v_{i+3})$ \textbf{or} $n = 6$}
            \State $\alpha \gets 0$ \Comment{This indicates a 2-chain of length 3 or a $P_{6}$}
        \Else
            \State $\alpha \gets 1$
        \EndIf
        \State $\beta \gets 1$
    \Else
        \State Color $v_{i}$ with existing color $C^{\star}$
        \State $\mathcal{F} \gets \{\mathcal{F}\cup C^{\star}\}$
        \State $\alpha \gets 1$
    \EndIf
\Else
    \State Color $v_{i}$ uniquely with color $C_{|\mathcal{C}|}$
    \State $\mathcal{C} \gets \{\mathcal{C}\cup C_{|\mathcal{C}|}\}$
    \State $\mathcal{F} \gets \{\mathcal{F}\cup C_{|\mathcal{C}|}\}$
    \State $\alpha \gets 0$
\EndIf
\EndFor
\State \textbf{return} $|\mathcal{C}|$ \Comment{The Dominator Chromatic Number is: $\chi_{d}(P_{n}) = |\mathcal{C}|$}
\State \textbf{return} $\mathcal{F}$ \Comment{The coloring of $V(P_{n})$, i.e., $\mathcal{F}_{i}=c(v_{i})\ \forall\ v_{i}\in V(P_{n})$}
\end{algorithmic}
\end{algorithm}

Now we show that the time complexity of the algorithm is $\mathcal{O}(n)$. This result follows since we may count a constant number of things the algorithm needs to check for each vertex, hence the algorithm takes at most $cn$ steps to complete for some $c\in\mathbb{N}$.

Next we show that the algorithm provides a minimum dominator coloring of any oriented path. 

\begin{theorem}\label{t-min}
The Minimum Dominator Coloring Algorithm results in a minimum dominator coloring of every oriented path.
\end{theorem}
\begin{proof}
Since any vertex with in-degree equal to zero is assigned to the same color class by this algorithm, no counterexample can come from these vertices. Since every vertex with an in-neighbor of out-degree one is colored uniquely by this algorithm, no counterexample can come from these vertices either. Thus any possible counterexample must come from the set of vertices whose entire in-neighborhood consists of vertices of out-degree two.

Since the algorithm colors each 2-chain optimally, it suffices to show that if each 2-chain is colored optimally, with the possible exception of 2-chains of length three whenever $C^{\star}$ already has been used, and the color $C^{\star}$ is common to all 2-chains which have non-uniquely colored vertices, all vertices of in-degree zero are colored similarly, and all vertices with an in-neighbor of out-degree one are colored uniquely (i.e., the conditions of the first paragraph of this proof are met), that the dominator coloring is minimum.

Clearly the color $C^{\star}$ must be common to all 2-chains (with non-uniquely colored vertices), else there are at least two color classes that can be combined. The MDC algorithm for orientations of paths does ensure that the color class $C^{\star}$ is the unique color class shared by 2-chains. With this established, it follows that each 2-chain must be (and in fact is) minimally dominator colored, except for possibly some 2-chains of length three which use the color $C^{\star}$ which was established prior to the instance of that 2-chain, as the only remaining vertices in our path are those vertices of 2-chains which must be colored uniquely. That the resulting dominator coloring is minimum immediately follows.
\end{proof}

\section{Conclusion}
In this paper we established the first algorithm which provides a minimum dominator coloring of directed graphs. Specifically, this algorithm provides a minimum dominator coloring for orientations of paths. We proved that this algorithm always results in a minimum dominator coloring of an oriented path, as well as that the algorithm runs in $\mathcal{O}(n)$ time.

The most likely extensions of this algorithm are to orientations of trees and cycles. While results on the dominator chromatic number of orientations of trees exist, specific results for this class of digraphs are not as established as they are for orientations of cycles, a class of graphs for which the dominator chromatic number is entirely determined. However, the acyclic structure of trees (and theoretically of directed acyclic graphs as well) lend them to being possible candidates for easy extensions of this algorithm.

To conclude this paper we mention an application of this result. For those interested, a python script, which uses this algorithm and provides a visual of a user selected orientation of a path, can be found online at the repository \url{https://github.com/cat-astrophic/MDC-orientations_of_paths/}.

Please feel free to use and modify this script to fit your needs!

\bibliography{MDCalgo}
\end{document}